\documentclass[conference,letterpaper]{IEEEtran}
\newcommand\blfootnote[1]{%
  \begingroup
  \renewcommand\thefootnote{}\footnote{#1}%
  \addtocounter{footnote}{-1}%
  \endgroup
}

\usepackage{amsmath,amssymb,amsfonts,dsfont,amsthm}
\usepackage{braket}
\usepackage{algorithmic}
\usepackage{graphicx}
\usepackage{textcomp}
\usepackage{xcolor}
\usepackage{physics}

\usepackage{flushend}

\newtheorem{theorem}{Theorem}[section]
\newtheorem{corollary}[theorem]{Corollary}
\newtheorem{lemma}[theorem]{Lemma}
\newtheorem{definition}[theorem]{Definition}

\newlength{\blank}
\newcommand{\1}{\mathds{1}}

\newcommand{\cP}{\mathcal{P}}

\newcommand{\cC}{\mathcal{C}}

\newcommand{\cM}{\mathcal{M}}
\newcommand{\cN}{\mathcal{N}}

\newcommand{\cS}{\mathcal{S}}
\newcommand{\cT}{\mathcal{T}}

\newcommand{\cX}{\mathcal{X}}
\newcommand{\cY}{\mathcal{Y}}

\newcommand{\ox}{\otimes}

\usepackage[style=ieee, backend=biber, sorting=none]{biblatex}

\begin{document}

\title{Quantum Hypothesis Testing Lemma for Deterministic Identification over Quantum Channels}

\author{\IEEEauthorblockN{Pau Colomer\textsuperscript{1}}
\IEEEauthorblockA{\textit{Chair of Theoretical}\\
\textit{Information Technology}\\
\textit{Technische Universit\"at}\\
\textit{M\"unchen, Germany}\\
pau.colomer@tum.de}
\and
\IEEEauthorblockN{Christian Deppe\textsuperscript{2}}
\IEEEauthorblockA{\textit{Institute for Communications}\\
\textit{Technology}\\
\textit{Technische Universit\"at}\\
\textit{Braunschweig, Germany}\\
christian.deppe@tu-bs.de}
\and
\IEEEauthorblockN{Holger Boche\textsuperscript{2,3}}
\IEEEauthorblockA{\textit{Chair of Theoretical}\\
\textit{Information Technology}\\
\textit{Technische Universit\"at}\\
\textit{M\"unchen, Germany}\\
boche@tum.de}
\and
\IEEEauthorblockN{Andreas Winter\textsuperscript{1,4}}
\IEEEauthorblockA{\textit{Department Mathematik/Informatik}\\
\textit{Abteilung Informatik}\\
\textit{Universit\"at zu}\\ 
\textit{K\"oln, Germany}\\
andreas.winter@uni-koeln.de}
}

\maketitle

\begin{abstract}
In our previous work, we presented the \emph{Hypothesis Testing Lemma}, a key tool that establishes sufficient conditions for the existence of good deterministic identification (DI) codes for memoryless channels with finite output, but arbitrary input alphabets. In this work, we provide a full quantum analogue of this lemma, which shows that the existence of a DI code in the quantum setting follows from a suitable packing in a modified space of output quantum states. 
Specifically, we demonstrate that such a code can be constructed using product states derived from this packing. 
This result enables us to tighten the capacity lower bound for DI over quantum channels beyond the simultaneous decoding approach. In particular, we can now express these bounds solely in terms of the Minkowski dimension of a certain state space, giving us new insights to better understand the nature of the protocol, 
and the separation between simultaneous and non-simultaneous codes. We extend the discussion with a particular channel example for which we can construct an optimum code.
\end{abstract}


\blfootnote{{\textsuperscript{1}Institute for Advanced Study, TUM, Garching, Germany.}\vspace{1pt}

{\textsuperscript{2}6G-life, 6G research hub,
Braunschweig and M\"unchen,
Germany.}\vspace{1pt}

{\textsuperscript{3}Munich Quantum Valley and MCQST, M\"unchen, Germany.}\vspace{1pt}

{\textsuperscript{4}ICREA {\&} Universitat Aut\`onoma de Barcelona, Barcelona, Spain.}
}


\section{Introduction}
Shannon's classical model of communication \cite{Shannon:TheoryCommunication} can be generalized to a quantum setting where a receiver aims to reconstruct a classical message encoded in a quantum state and sent through $n$ uses of a memoryless quantum channel. In such a setting, the number of distinct messages $M$ that can be reliably transmitted -- that is, with arbitrarily small probability of error for sufficiently long codewords -- grows exponentially with $n$, leading to the linear relation $\log M\sim Rn$, with $R=\frac1n \log M$ the $\emph{rate}$ of the code, and $\log$ the binary logarithm. The supremum of all achievable rates in the asymptotic limit of large $n$ is known as the \emph{capacity} $C$ and is given by the celebrated Holevo-Schumacher-Westmoreland (HSW) theorem and some subsequent refinements \cite{holevo:capacity,SW:capacity,Winter:Strong_converse,ON:Strong_converse}.

This setting contrasts with the \emph{identification} problem, where the receiver is not required to recover the sent message, but only to check whether it matches a specific message of his interest. The decoding is therefore limited to a hypothesis test, reducing the communication complexity of the task \cite{Ja:ID_easier}. A full characterization of identification over classical discrete memoryless channels (DMCs) was first provided by Ahlswede and Dueck \cite{AD:ID_ViaChannels}, and later generalized to the quantum setting by L\"ober \cite{loeber:PhDThesis}. Their results show that the number $N$ of identifiable messages can grow exponentially faster with the number $n$ of uses of the channel than in transmission, i.e., $\log N\sim2^{nR}$. This means that we can identify exponentially more messages than we can transmit.

Interestingly, the exponential increase of the code size that identification codes offer is lost if deterministic encodings (without any use of randomness) are imposed. As a matter of fact, it was already noted in \cite{AD:ID_ViaChannels} that DI over any DMC could only achieve a linear scaling with $n$ (the same as in transmission). 
Deterministic codes, however, are in general easier to implement, simulate, and to explicitly construct \cite{DI_simpler_impl,DI_explicit_construction,VDTB:practical-DI}; so, in spite of the poorer performance, an interest was kept in DI codes due to potential applications. This interest was renewed when it was proven that certain channels with continuous input can have the superlinear (so-called \emph{linearithmic}) scaling $\log N\sim Rn\log n$ of the code length \cite{SPBD:DI_power,DI-poisson_mc}. I.e., it is better than transmission and DI over DMCs, which can only achieve a linear scaling, but worse than general randomized identification, which can achieve an exponential scaling of the code length $\log N$.

In recent work \cite{CDBW:DI_proceedings,CDBW:DI_classical,CDBW:Reliability_ICC,CDBW:reliability_Arxiv}, it has been proven that this superlinear behaviour exhibited in DI codes over some channels is in fact a general feature of deterministic identification over channels $W$ with discrete output $\cY$ but arbitrary input alphabets $\cX$, and that the corresponding capacities are related to the fractal dimension of a certain modification of the output probability set. The code construction part of these results is based on the so-called Hypothesis Testing Lemma, which shows that the entropy typical set $\cT_{x^n}^\delta$ for each possible code word $x^n$, and $0<\delta\leq\sqrt{n}\log|\cY|$ can be a decent surrogate of the optimal identification test: 

\begin{lemma}[{Colomer~\emph{et~al.}~\cite{CDBW:DI_classical}}]
\label{lemma:HTL}
Consider $n$ uses of the classical channel $W:\cX\rightarrow\cY$ and two points $x^n,{x'}^n \in \mathcal{X}^n$ such that $1-\frac12\left\|W_{x^n}-W_{{x'}^n}\right\|_1 \leq \epsilon$. Then,
\[\begin{split}
  W_{{x'}^n}(\mathcal{T}_{x^n}^\delta) 
   &\leq 2\exp\left(-\delta^2/36K(|\mathcal{Y}|)\right)\\
   &\quad+ \epsilon\left(1 + 2^{2\delta\sqrt{n}}2^{H(W_{x^n})-H(W_{{x'}^n})}\right),
\end{split}\]
with $K(d) = (\log\max\{d,3\})^2$, $W_{x^n}$ the output probability distribution of the input $x^n$, and $H(W_{x^n})$ the corresponding Shannon entropy. 
\hfill\qed
\end{lemma}

From this, it is shown that a good classical DI code can be constructed just from a packing in the output probability space with large enough distance between its elements \cite[Thm.~III.2]{CDBW:DI_classical}. The performance of such code can be then studied to calculate an achievability (lower) bound for the classical DI capacity \cite[Thm.~V.7]{CDBW:DI_classical}. These results are not directly generalizable to the quantum setting as Lemma \ref{lemma:HTL} is not available for quantum states. However, through a couple of tricks, the classical lemma can be generalised to the quantum case. Let us prepare the ground 
with some definitions.

\begin{definition}
\label{def:qID}
An $(n,N,\lambda_1,\lambda_2)$-\emph{DI code} for a classical-quantum channel $W:\cX\rightarrow\cS(B)$ is a collection $\{(u_j,E_j):j\in[N]\}$ of code words $u_j\in\cX^n$ and POVM elements $0\leq E_j\leq\1$ on $B^n$, such that for all $j\neq k\in[N]$
\[
  \Tr W_{u_j}^n E_j \geq 1-\lambda_1,
  \quad 
  \Tr W^n_{u_j} E_k \leq \lambda_2,
\]
where $W^n_{u_j} = W_{x_1}\otimes\cdots\otimes W_{x_n} \in\cS(B^n)$ is the output quantum state when the input sequence is $u_j=x_1\ldots x_n\in\cX^n$. 

The largest $N$ such that an $(n,N,\lambda_1,\lambda_2)$-DI code exists is denoted $N_{\text{DI}}(n,\lambda_1,\lambda_2)$.
\end{definition}

Notice that, while in quantum transmission the decoding consists of a fixed POVM, the outcome of which is used to guess the message that was sent; in identification, the measurement the receiver applies is dependent on the message $j\in[N]$ that he wants to identify, for which the hypothesis test $(E_j,\1-E_j)$ is going to serve as decoder. The different effects $E_j$ do not necessarily need to be the elements of a POVM. In fact, due to the possibility of complementarity of quantum observables, 
in general the tests $(E_j,\1-E_j)$ are not even simultaneously measurable.

However, if we insist that they are, meaning that there is a POVM $(D_m:m\in\cM)$ on $B^n$ such that each $(E_j,\1-E_j)$ is a coarse-graining of it, i.e.~there is a subset $\cM_j\subset\cM$ with $E_j = \sum_{m\in\cM_j} D_m$, the DI code is called \emph{simultaneous} \cite{loeber:PhDThesis}. The largest code size is then denoted $N_{\text{DI}}^{\text{sim}}(n,\lambda_1,\lambda_2)$. 

The main finding of \cite{CDBW:DI_classical} is that the message length $\log N$ of DI codes generically scales as $n\log n$, thus motivating the definition  of the slightly superlinear DI capacity as 
\[
  \dot{C}_{\text{DI}}(W) 
    := \inf_{\lambda_1,\lambda_2>0} \liminf_{n\rightarrow\infty} 
    \frac{1}{n\log n}\log N_{\text{DI}}(n,\lambda_1,\lambda_2),
\]
and similarly for the simultaneous DI capacity $\dot{C}_{\text{DI}}^{\text{sim}}(W)$.


There is a simple way of constructing a simultaneous DI code for the cq-channel $W$: choose a POVM $T=(T_y:y\in\cY)$ on $B$ and concatenate $W$ with it, obtaining a classical channel $\overline{W}(y|x)=\Tr W_xT_y$. Using the same measurement $T$ on each of the $n$ channel uses transforms the cq-channel $W^n$ into the classical i.i.d.~channel $\overline{W}^n$, so Lemma \ref{lemma:HTL} can be applied, a good DI code constructed, and the slightly superlinear simultaneous DI capacity of $W$ found to be lower-bounded by the slightly superlinear DI capacity of $\overline{W}$.

\begin{theorem}[{Colomer~\emph{et~al.}~\cite{CDBW:DI_classical}}]
\label{thm:cq-DI}
With the previously defined channels, let $\widehat{\cX}=\overline{W}(\cX)\subset\cP(\cY)$, $\widetilde{\cX}=W(\cX)\subset\cS(B)$. Then, the slightly superlinear simultaneous and general DI capacities of a cq-channel $W:\cX\rightarrow\cS(B)$ are bounded as follows (for any $\lambda_1,\lambda_2>0$ with $\lambda_1+\lambda_2<1$):
\[
\frac14{\underline{d}_M\!\left(\sqrt{\!\widehat{\cX}}\right)}
   \leq \dot{C}_{\text{DI}}^\text{sim}(W) 
   \leq \dot{C}_{\text{DI}}(W)
   \leq\frac12 \underline{d}_M\!\left(\!\sqrt{\!\widetilde{\cX}}\right),
\]
with $\sqrt{\!\widetilde{\cX}} = \left\{ \sqrt{\sigma} : \sigma\in\widetilde{\cX} \right\}$ and $\underline{d}_M$ the lower Minkowski dimension (see the preliminaries in Section \ref{sec:preliminaries}).
\end{theorem}

The simultaneous lower bound is of course automatically a lower bound to the general case. In contrast, the upper bound is valid for arbitrary DI codes. In the present work, we develop (in Section \ref{sec:lemma}) a quantum analogue to the Hypothesis Testing Lemma \cite[Lemma~III.1]{CDBW:DI_classical}, and we use it (in Section \ref{sec:qDI}) to construct a code without the need of a fixed POVM at the output, that is, without the simultaneity restriction. From this construction, we can directly find a tighter lower bound on the non-simultaneous DI capacity. Specifically, we prove
\begin{equation}
  \label{eq:final_bounds}
\frac14\underline{d}_M\!\left(\!\sqrt{\!\widetilde{\cX}}\right)
   \leq \dot{C}_{\text{DI}}(W)
   \leq \frac12\underline{d}_M\!\left(\!\sqrt{\!\widetilde{\cX}}\right),
\end{equation}
showing that both the upper and lower bound can be written in terms of the same dimension of the output state space. This result also sheds some new light on the possible separation of simultaneous and non-simultaneous DI capacities.

These results for cq-channels can be extended to fully quantum channels under the restriction that only product states are used to encode the messages. Specifically, for a quantum channel $\cN:A\rightarrow B$ we fix a subset $\cX\subset\cS(A)$ of input quantum states, and define an $(n,N,\lambda_1,\lambda_2)$-DI $\cX$-code as a collection $\{(\rho_j=\rho_j^{(1)}\otimes\dots\otimes\rho_j^{(n)},E_j):j\in[N]\}$ with $\rho_j^{(i)}\in\cX$, such that $\Tr\cN^n_{\rho_j}E_j\geq 1-\lambda_1$ and $\Tr\cN^n_{\rho_k}E_j\leq \lambda_2$ for all $j\neq k\in[N]$. Then, the capacity bounds of an $\cX$-code are the same as in Eq.~\eqref{eq:final_bounds}, with $\widetilde\cX:=\cN(\cX)$.

Finally, in Section \ref{sec:example} we study a particular case of a channel which allows a simple code construction bypassing the use of the quantum hypothesis testing lemma. We proof that our construction is optimal and attains the upper bound on the general capacity formula \eqref{eq:final_bounds}, meaning that, for this particular channel $W$, 
$\dot{C}_\text{DI}(W) = \frac{1}{2} \underline{d}_M\!\left(\widetilde{\cX}\right)$. We finish with a discussion in Section \ref{sec:conclusions}.



\section{Preliminaries}
\label{sec:preliminaries}
We consider a cq-channel $W:\cX \rightarrow \cS(B)$ with a finite-dimensional Hilbert space $B$ and an arbitrary measure space $\cX$, meaning that $W$ is a measurable map. The output states are denoted $W_x\in\cS(B)$ for $x\in\cX$. For $n$ copies, $W^n:\cX^n \rightarrow \cS(B^n)$ maps $x^n=x_1\ldots x_n$ to $W_{x^n} = W_{x_1}\ox\cdots\ox W_{x_n}$.

\begin{definition}
The \emph{entropy typical projector} $\Pi_{x^n}^\delta$ of $W_{x^n}$ and for $\delta>0$ is defined as
\[
  \Pi_{x^n}^\delta := \left\{ S(W_{x^n})-\delta\sqrt{n} \leq -\log W_{x^n} \leq S(W_{x^n})+\delta\sqrt{n} \right\} \!, 
\] 
signifying the projector onto the span of all eigenspaces of $-\log W_{x^n}$ corresponding to eigenvalues in the interval $\left[S(W_{x^n})-\delta\sqrt{n}; S(W_{x^n})+\delta\sqrt{n}\right]$.
\end{definition}

The properties of the typical projector that we shall need are these: first, by definition, $\Pi_{x^n}^\delta$ commutes with $W_{x^n}$, and the \emph{truncated state} satisfies
\begin{equation}
  \label{eq:truncated}
  \widetilde{W}_{x^n} := \Pi_{x^n}^\delta W_{x^n} \Pi_{x^n}^\delta \leq W_{x^n}.
\end{equation}
Next (cf.~\cite[Lemma~II.1]{CDBW:DI_classical}),
\begin{equation}
  \label{eq:typicality}
  \Tr W_{x^n} \Pi_{x^n}^\delta \geq 1-2\exp\left( -\delta^2/36K(|B|) \right),
\end{equation}
where $K(d) = (\log\max\{d,3\})^2$ and both $\log$ and $\exp$ are understood to be to base $2$. Finally, again by definition, 
\begin{equation}
  \label{eq:AEP}
  2^{-S(W_{x^n})-\delta\sqrt{n}} \Pi_{x^n}^\delta 
     \leq \widetilde{W}_{x^n}
     \leq 2^{-S(W_{x^n})+\delta\sqrt{n}} \Pi_{x^n}^\delta.
\end{equation}

As already commented in the introduction, the capacity bounds are written in terms of the \emph{Minkowski dimension} (also known as covering, Kolmogorov or entropy dimension \cite{Falconer:fractal}) of the square-root output space. The Minkowski dimension of a set $F$ is defined as a ratio between the number of elements that you need to create a covering $\Gamma_\delta$ (or a packing $\Pi_\delta$) of $F$ and the size of each element $\delta$ in the limit of $\delta\rightarrow0$. It is quite common that this limit does not exist, so we define the upper and lower Minkowski dimensions through the superior and inferior limits as follows:
\begin{align*}
 \overline{d}_M(F) &:= \limsup_{\delta\rightarrow0}\frac{\log \Gamma_\delta(F)}{-\log\delta} 
  = \limsup_{\delta\rightarrow0} \frac{\log \Pi_\delta(F)}{-\log\delta}, \\
 \underline{d}_M(F) &:= \liminf_{\delta\rightarrow0}\frac{\log \Gamma_\delta(F)}{-\log\delta}
  = \liminf_{\delta\rightarrow0} \frac{\log \Pi_\delta(F)}{-\log\delta}.
\end{align*}
Intuitively, this is an object that captures the 
geometric complexity of the set $F$ at small scales. This is especially important in intrinsically irregular (fractal) sets, whereas for smooth objects (such as manifolds) the Minkowski dimension coincides with the topological one.

\section{Quantum Hypothesis Testing Lemma}
\label{sec:lemma}
We prove in this section the quantum analogue of \cite[Lemma~III.1]{CDBW:DI_classical}.
It is worth noting here, that this is not at all a trivial extension from its classical counterpart. There, the proof is based on the study of the intersection between different subsets that could serve as identification tests. In the quantum case, however, the tests are no longer subsets of classical strings, but quantum projectors, so some classical technical tools (for example intersections and unions) are no longer available. Instead, we use ideas from \cite{Hayashi_2002} to work out a quantum version which, while inspired by its classical sister and having a similar aspect (with only one slightly worse constant factor), is (and has to be) different in its methods. 
\begin{lemma}
\label{lemma:qHTL}
For $n$ uses of the cq-channel $W$, consider two points $x^n,\,{x'}^n \in \cX^n$ such that $1-\frac12\|W_{x^n}-W_{{x'}^n}\|_1 \leq \epsilon$. Then, 
\[
\begin{split}
\Tr W_{{x'}^n} \Pi_{x^n}^\delta 
    \leq & \,\, 2\exp\left(-\delta^2/36K(|B|)\right)\\ 
         &+ 2\epsilon\left( 1 + 2^{2\delta\sqrt{n}}2^{S(W_{x^n})-S(W_{{x'}^n})} \right).
\end{split}
\]
\end{lemma}

\begin{proof}
Let $S$ be the Helstrom optimal projector for the discrimination of the equiprobable hypotheses $W_{x^n}$ and $W_{{x'}^n}$. By the assumption on trace distance, we have
\begin{equation}\label{eq:tr_properties}
  \Tr W_{x^n}(\1-S) \leq \epsilon, \quad \Tr W_{{x'}^n} S \leq \epsilon. 
\end{equation}
Furthermore, we will use the pinching inequality \cite{Hayashi_2002} (see also \cite{Bhatia_2000}), which for the measurement $(S,\1-S)$ gives
\begin{align}
\nonumber \Pi_{x^n}^\delta &\leq 2 S\Pi_{x^n}^\delta S + 2(\1-S)\Pi_{x^n}^\delta(\1-S)\\
                 &\leq 2S + 2(\1-S)\Pi_{x^n}^\delta(\1-S).\label{eq:pinching}
\end{align}

Let us start by dividing the trace as follows:
\begin{equation}\label{eq:split}
\Tr W_{{x'}^n}\Pi_{x^n}^\delta =    \Tr \widetilde{W}_{{x'}^n}\Pi_{x^n}^\delta 
           + \Tr\left(W_{{x'}^n}-\widetilde{W}_{{x'}^n}\right)\Pi_{x^n}^\delta.
\end{equation}
Now, introducing $\eta = 2\exp\left(-\delta^2/36K(|B|)\right)$ as an abbreviation, and using Eq.~\eqref{eq:truncated} first and Eq.~\eqref{eq:typicality} after, we observe that the second term is just:
\[
\begin{split}
\Tr\!\left(W_{{x'}^n}-\widetilde{W}_{{x'}^n}\!\right)\!\Pi_{x^n}^\delta
    &\!\leq\! \Tr W_{{x'}^n}\!-\!\widetilde{W}_{{x'}^n}=1-\Tr W_{{x'}^n}\Pi_{x^n}^\delta\\
    &\!\leq1-1+\eta=\eta.
\end{split}
\]

For the first term in Eq.~\eqref{eq:split} we use the pinching inequality [Eq.~\eqref{eq:pinching}] first, and the trace properties of $S$ [Eq.~\eqref{eq:tr_properties}] after:
\begin{align}
\Tr \widetilde{W}_{{x'}^n}\Pi_{x^n}^\delta
\nonumber    &\leq 2\Tr \widetilde{W}_{{x'}^n} S + 2\Tr \widetilde{W}_{{x'}^n} (\1-S)\Pi_{x^n}^\delta(\1-S)\\ 
    &\leq 2\epsilon + 2\Tr \widetilde{W}_{{x'}^n} (\1-S)\Pi_{x^n}^\delta(\1-S).\label{eq:step}
\end{align}

Observe now that, from Eqs.~\eqref{eq:truncated} and \eqref{eq:tr_properties}, we can bound
\[
\begin{split}
\Tr W_{x^n} \Pi_{x^n}^\delta(\1-S)\Pi_{x^n}^\delta 
    &=\Tr \widetilde{W}_{x^n} (\1-S)\\ 
    &\leq \Tr W_{x^n}(\1-S) \leq \epsilon,
\end{split}
\]
and since $\Pi_{x^n}^\delta(\1-S)\Pi_{x^n}^\delta$ is supported in the typical subspace of $W_{x^n}$, using the lower bound in Eq.~\eqref{eq:AEP}, the cyclicity of the trace, and the idempotency of projectors, we find 
\[
\begin{split}
\Tr \Pi_{x^n}^\delta(\1-S)\Pi_{x^n}^\delta 
    &\leq 2^{S(W_{x^n})+\delta\sqrt{n}}\Tr W_{x^n} \Pi_{x^n}^\delta (\1-S)\Pi_{x^n}^\delta\\ 
    &\leq \epsilon\cdot2^{S(W_{x^n})+\delta\sqrt{n}}.
\end{split}
\]
Also, notice $\Tr \Pi_{x^n}^\delta(\1-S)\Pi_{x^n}^\delta = \Tr (\1-S)\Pi_{x^n}^\delta(\1-S)$, and so, using the upper bound in Eq.~\eqref{eq:AEP}, we can write
\[\begin{split}
\Tr \widetilde{W}_{{x'}^n} &(\1-S)\Pi_{x^n}^\delta(\1-S)\\
   &\leq 2^{-S(W_{x'^n})+\delta\sqrt{n}}\Tr(\1-S)\Pi_{x^n}^\delta\Pi_{x^n}^\delta(\1-S)\\
   &\leq\epsilon\cdot 2^{S(W_{x^n})+\delta\sqrt{n}}2^{-S(W_{{x'}^n})+\delta\sqrt{n}}.
\end{split}\]
Substituting this onto Eq.~\eqref{eq:step} and all that into Eq.~\eqref{eq:split} we conclude the proof:
\[
 \Tr W_{{x'}^n} \Pi_{x^n}^\delta 
    \leq \eta + 2\epsilon\left( 1 + 2^{2\delta\sqrt{n}}2^{S(W_{x^n})-S(W_{{x'}^n})} \right).\vspace{-0.65cm}
\]
\end{proof}

\section{Identification via quantum channels}\label{sec:qDI}
As we have discussed in the introduction, the Quantum Hypothesis Testing Lemma \ref{lemma:qHTL} allows us to construct non-simultaneous DI codes for cq-channels, which permit to improve the achievability bounds for the unrestricted DI capacity. We provide the proves in this section, together with the extension to fully quantum channels with product state encodings and optimistic capacities.  

\begin{theorem}\label{thm:qDI_construction}
    For $n$ uses of the cq-channel $W$, let $0<\delta\leq\sqrt{n}\log|B|$, and assume $N$ codewords $u_j\in\cX^n$ with pairwise distance $1-\frac12\|W_{u_j}-W_{u_k}\|_1 \leq 2^{-3\delta\sqrt{n}}$, i.e.~the $W_{u_j}$ form a packing in $\cS(B^n)$. Then there is a subset $\mathcal{C}\subset[N]$ of cardinality $|\mathcal{C}|\geq N/\left\lceil n\log|B|\right\rceil$ such that the collection $\left\{\left(u_j,E_j=\Pi_{u_j}^\delta\right) : j\in\mathcal{C}\right\}$ is an $(n,|\cC|,\lambda_1,\lambda_2)$-DI code, with $\lambda_1 = \eta$ and $\lambda_2 =\eta + 6\exp\left(-\delta\sqrt{n}\right)$.
\end{theorem}
\begin{proof}
The proof is analogous to the classical one (c.f.~\cite[Thm.~III.2]{CDBW:DI_classical}). We start by dividing $[N]$ into $S=\left\lceil n\log|B|\right\rceil$ parts $\mathcal{C}_s$ ($s=1,\ldots,S$) in such a way that all $j\in\mathcal{C}_s$ have $S(W_{u_j})\in[s-1;s]$. This is possible since the von Neumann entropies of the states $W_{x^n}$ are in the interval $[0;n\log|B|]$.

Then, we choose $\mathcal{C}$ as the largest $\mathcal{C}_s$, which satisfies the cardinality lower bound $|\mathcal{C}|\geq N/\left\lceil n\log|B|\right\rceil$ by the pigeonhole principle. We bound the first kind of error of through Eq.~\eqref{eq:typicality} and, since within $\mathcal{C}$, the entropies $S(W_{u_j})$ differ by at most $1$ from each other, we can apply Lemma \ref{lemma:qHTL} to bound the second kind of error as claimed.
\end{proof}

The theorem above tells us that if we create a packing with large enough distance between output quantum states, then we can construct a deterministic identification code with the described characteristics. We stress that the described code is not simultaneous, as in general the entropy-typical projectors are not necessarily simultaneously measurable. 

Following closely the classical capacity proofs in \cite{CDBW:DI_classical}, we can now use Theorem \ref{thm:qDI_construction} to improve the capacity bounds for DI over a cq-channel.

\begin{theorem}
\label{thm:cq-DI-capacity}
The slightly superlinear DI capacity of a cq-channel $W:\cX\rightarrow\cS(B)$ is bounded as follows (for any $\lambda_1,\lambda_2>0$ with $\lambda_1+\lambda_2<1$):
\[\begin{split}
\frac14{\underline{d}_M\!\left(\!\sqrt{\!\widetilde{\cX}}\right)}
   \leq \dot{C}_{\text{DI}}(W) 
   &\leq \liminf_{n\rightarrow\infty} \frac{1}{n\log n} \log N_{\text{DI}}(n,\lambda_1,\lambda_2)\\
   &\leq\frac12 \underline{d}_M\!\left(\!\sqrt{\!\widetilde{\cX}}\right),
\end{split}\]
where $\widetilde{\cX}=W(\cX)\subset\cS(B)$ and $\sqrt{\!\widetilde{\cX}} = \left\{ \sqrt{\sigma} : \sigma\in\widetilde{\cX} \right\}$.
\end{theorem}
\begin{proof}
We only need to prove the lower bound, as the upper bound has been shown in \cite{CDBW:DI_classical}. We start with a packing $\cX_{n^{-\alpha}}\subset\cX$, that is a set of points such that for all $x\neq x'\in\cX_{n^{-\alpha}}$, $|\sqrt{W_x}-\sqrt{W_{x'}}|_2\geq n^{-\alpha}$. By definition of the lower Minkowski dimension $d=\underline{d}_M\!\left(\!\sqrt{\!\widetilde{\cX}}\right)$, for every $\epsilon > 0$ there is a positive constant $K>0$ such that we can find such a packing with $\lvert{\cX_{n^{-\alpha}}}\rvert \geq (Kn^{\alpha})^{d-\epsilon}$ elements. 

We aim to apply Theorem \ref{thm:qDI_construction}, so we need to bound the trace distance between output distributions. In blocklength $n$, this distance can be related to the following sum of letterwise Euclidean distances \cite{CDBW:DI_classical}: 
\begin{equation}
\label{eq:Trdist->Eucl}
  -\ln\left( 1-\frac12\|W_{x^n}-W_{{x'}^n}\|_1 \right)
  \geq \frac14 \sum_{i=1}^n \left|\sqrt{W_{x_i}} -\sqrt{W_{x_i'}}\right|_2^2.
\end{equation}
By construction, if we let $x^n,{x'}^n\in\cX_{n^{-\alpha}}^n$, each term in the latter sum is either $0$ or at least $n^{-2\alpha}$. Let us now choose a code $\cC_t\subset\cX_{n^{-\alpha}}^n$ with a minimum Hamming distance between its elements, that is $d_H(x^n,{x'}^n) \geq tn$ for all $x^n\neq {x'}^n\in\cC_t$, where $t>0$ is an arbitrarily small constant. We immediately see that for code words in $\cC_t$, the sum in Eq.~\eqref{eq:Trdist->Eucl} will have non-zero elements in at least $tn$ positions, and for each one of those we can lower-bound the norm squared by $n^{-2\alpha}$, so
\begin{equation}
\label{eq:Main_achiev_exponent}
\!\!\!\!-\ln\!\left(\!1-\frac12\|W_{x^n}-W_{{x'}^n}\|_1\!\right)\!\geq\!\frac14 \!\left(n^{-2\alpha}\right)tn=\frac{t}{4}n^{1-2\alpha}.
\end{equation}
Now, to satisfy the conditions of Theorem \ref{thm:qDI_construction} which constructs an identification code from a packing on the output probability set, we require a scaling of at least $\sqrt{n}$ in the exponent of Eq.~\eqref{eq:Main_achiev_exponent}. So, we need $1-2\alpha\geq\frac12$, i.e.~$\alpha\leq\frac14$, and then, by virtue of Theorem \ref{thm:qDI_construction}, it is possible to find asymptotically good DI codes having $N = \left\lfloor \lvert{\cC_t}\rvert/(n\log\lvert{B}\rvert) \right\rfloor$ messages. It is only left to bound the size of the code $\cC_t$. By elementary combinatorics, the Hamming ball around any point in $\cX_{n^{-\alpha}}^n$ of radius $tn$ has 
$\leq \binom{n}{tn} |\cX_{n^{-\alpha}}|^{tn}$ elements. Hence, any maximal code of Hamming distance $tn$ has a number of code words at least the ratio between the total number of elements and the size of the Hamming ball (otherwise the code could be extended):
\begin{equation}\label{eq:size_Ct}
\lvert{\cC_t}\rvert \geq 
  \frac{\lvert{\cX_{n^{-\alpha}}}\rvert^n}{\binom{n}{tn}\lvert{\cX_{n^{-\alpha}}}\rvert^{tn}}
  \geq    2^{-n} \lvert{\cX_{n^{-\alpha}}}\rvert^{n(1-t)}.
\end{equation}

Finally, substituting the bound $|\cX_{n^{-\alpha}}| \geq \left({Kn^\alpha}\right)^{d-\epsilon}$ we obtain
\begin{align}
\nonumber N \geq \frac{\lvert{\cC_t}\rvert}{n\log|B|}
   &\geq 2^{-n-\log n-\log\log|B|} \left[\left(Kn^\alpha\right)^{d-\epsilon}\right]^{n(1-t)}\\
   &\geq 2^{\alpha(1-t)(d-\epsilon)n\log n -O(n)}. \label{eq:general_code_size}
\end{align}
Above we have inferred that any $\alpha<\frac14$ is suitable to guarantee asymptotically vanishing errors, and $\epsilon,t>0$ are arbitrarily small, showing that any $n\log n$-rate below $\frac14 d$ is attainable. Putting everything together we obtain 
\[
\begin{split}
\hspace{0.6cm}\dot{C}_\text{DI}(W)
    &= \inf_{\lambda_1,\lambda_2>0} \liminf_{n\rightarrow\infty}\frac{1}{n\log n} \log N_\text{DI}(n,\lambda_1,\lambda_2)\\
    &\geq \liminf_{n\rightarrow\infty}\frac{1}{n\log n} \log N
    \geq \frac14 \underline{d}_M\!\left(\!\sqrt{\!\widetilde{\cX}}\right).\hspace{0.6cm}\qed
\end{split}
\]
\renewcommand{\qed}{}
\end{proof}

In \cite{CDBW:DI_classical} the general DI-capacity over a cq-channel was lower bounded by the simultaneous DI-capacity. And the bound was expressed in terms of the dimension of a classical space resulting from the concatenation of the cq-channel and a given quantum measurement. Now, we can express the general DI-capacity bounds (both lower and upper) solely in terms of the Minkowski dimension of the (square root) space of output quantum states. Making it more tight, and showing the importance of such a dimension as the key feature when describing the capacity.

As described in the introduction, we can also apply this method to the parallel case of deterministic identification over fully quantum channels, $\cN:A\rightarrow B$ under the restriction that only product states are used in the encoding. 
Then:

\begin{corollary}
\label{cor:restricted-qq-channel}
For a quantum channel $\cN:A\rightarrow B$ and a product state restriction to encodings in $\cX\subset\cS(A)$, the slightly superlinear DI capacity is bounded as follows:
\[\begin{split}
  \frac14\underline{d}_M\!\left(\!\sqrt{\!\widetilde{\cX}}\right)
   \!\leq\! \dot{C}_{\text{DI}}(\cN\vert_{\cX}) 
   &\!\leq\! \liminf_{n\rightarrow\infty} \frac{1}{n\log n} \!\log\! N_{\text{DI},\cX}(n,\lambda_1,\lambda_2)\! \\
   &\!\leq \frac12\underline{d}_M\!\left(\!\sqrt{\!\widetilde{\cX}}\right).
   \hspace{2.4cm}\qed
\end{split}\]
\end{corollary}

Both Theorem \ref{thm:cq-DI-capacity} and Corollary \ref{cor:restricted-qq-channel} give bounds for the so-called \emph{pessimistic} capacity, which means that the capacity is a rate realised for all sufficiently large values of $n$, i.e.~as the limit inferior. It might of course be possible to find a diverging subsequence of block lengths for which the rate converges to a higher number. The largest such rate is the \emph{optimistic capacity}, defined as superior limit superior. The known non-simultaneous lower bounds for the optimistic capacities can be improved through our new Quantum Hypothesis Testing Lemma, similarly to the pessimistic case. 

\begin{theorem}
\label{thm:optimistic}
The slightly superlinear optimistic DI capacity of a cq-channel $W:\cX\rightarrow\cS(B)$ is bounded as follows (for any $\lambda_1,\lambda_2>0$ with $\lambda_1+\lambda_2<1$):
\[
\begin{split}
  \frac14{\overline{d}_M\! \left(\!\sqrt{\!\widetilde{\cX}}\right)}
   \leq \dot{C}_{\text{DI}}^\text{opt}(W) 
   &\leq \limsup_{n\rightarrow\infty} \frac{1}{n\log n} N_{\text{DI}}(n,\lambda_1,\lambda_2) \\ 
   &\leq\frac12 \overline{d}_M\!\left(\!\sqrt{\!\widetilde{\cX}}\right).
\end{split}
\]
Similarly (cf.~Corollary \ref{cor:restricted-qq-channel}), 
\(
\frac14\overline{d}_M\!\left(\!\sqrt{\!\widetilde{\cX}}\right)
\leq\dot{C}_{\text{DI}}^{\text{opt}}(\cN\vert_{\cX})
\leq\frac12\overline{d}_M\!\left(\!\sqrt{\!\widetilde{\cX}}\right).
\) 
\hfill\qedsymbol
\end{theorem}

\section{Contrasting example: \\
optimality of the packing}\label{sec:example}
One of the most interesting open problems in both our classical and quantum DI settings is closing the gap between the upper and lower capacity bounds in the general case, which differ on the prefactor ($\frac{1}{4}$ in the lower bound and $\frac12$ in the upper bound). It was already noted in our previous work \cite{CDBW:DI_classical}, that if all we wanted was a packing in the output probability space, we could get a lower bound with prefactor $\frac12$, solving the problem. However, the application of either classical or quantum hypothesis testing lemma and the subsequent code construction, induces us to require a trace distance exponentially close to one, and this exponent is ruled by $\sqrt{n}$. This can clearly be seen in Eq.~\eqref{eq:Main_achiev_exponent} and the discussion that follows. There, we require the parameter $\alpha\leq1/4$, so that we can have the $\sqrt{n}$ exponent needed to apply Theorem \ref{thm:qDI_construction}. As we will see, an optimal packing for DI can only require a prefactor of $\frac12$. Furthermore, using such a packing, we will be able to construct a reliable code (with bounded errors that vanish for increasing $n$) for a particular classical-quantum channel $W:\cX\rightarrow B$ for which we can bypass the use of the quantum hypothesis testing lemma.

Let the cq-channel $W$ assign for each (classical) input letter $x\in\cX$, a pure state $W_x=\ketbra{\psi_x}{\psi_x}\in \cS_0(B)=\{\ketbra{\psi}{\psi}: \ket{\psi}\in B\}$. The output density matrix space is therefore a subset $\widetilde\cX\subset\cS_0(B)$. 

\begin{theorem}\label{thm:example}
The superlinear DI capacity of the cq-channel $W$ described above is $\dot{C}_\text{DI}(W)=\frac12\underline{d}_M\left(\widetilde{\cX}\right)$.
\end{theorem}
\begin{proof}
We define a code constructed by $N$ codewords $u_j=x_1\dots x_n\in\cX^n$  (with $j\in[N]$) and effect identification operators $E_j=\bigotimes_{i=1}^nW_{x_i}$. With this construction notice that, for all $j\in[N]$,
\[
\Tr W_{u_j} E_j 
 = \Tr W_{x^n} W_{x^n}
 = 1.
\]
Therefore, the first type of error is zero by definition. Indeed, $\lambda_1 = \sup_j\Tr W_{u_j}(\1 - E_j) = 1-\sup_j\Tr W_{u_j}(E_j)=0$. 

To bound the second kind of error we need to do a bit more work. However, we can retrace the same techniques used in the proof of Theorem \ref{thm:cq-DI-capacity}: we start with an $\cX_{n^{-\alpha}}$ packing, which by definition of the lower Minkowski dimension $d$ has size lower bounded by $|\cX_{n^{-\alpha}}|\geq(Kn^\alpha)^{d-\epsilon}$. Then we choose a code $\cC_t\subset\cX_{n^{-\alpha}}^n$ with minimum Hamming distance between its elements $d_H(x^n,x'^n)\geq tn$ for all $x^n\neq x'^n\in\cC_t$ with $t>0$ an arbitrarily small constant. With this construction --which up until now has followed the same steps as the previous proof-- we get exactly to Eq.~\eqref{eq:Main_achiev_exponent}:
\[
-\ln\left(1-\frac12\|W_{x^n}-W_{{x'}^n}\|_1\right)\geq\frac{t}{4}n^{1-2\alpha}.
\]
Now, however, as hinted in the introduction of this section, we have no restrictions on the value of $\alpha$ apart that we need to be able to bound the second type of error. We notice that we can choose it arbitrarily close to $\frac12$. Indeed, let $\gamma>0$ be an arbitrarily small constant, then we choose $\alpha=\frac{1-\gamma}{2}$. And therefore, with some algebra onto Eq.~\eqref{eq:Main_achiev_exponent} (reproduced above) we get
\[
\frac12\|W_{x^n}-W_{{x'}^n}\|_1\geq1-2^{-n^{\gamma}t/4}.
\]
We can use this to bound directly the second type of error. Let $u_j=x^n,u_k=x'^n$, then:
\[
\Tr W_{u_j}E_k \leq \Tr W_{u_k}E_k - (1-2^{-n^{\gamma}t/4}) = 2^{-n^{\gamma}t/4}.
\]
This is of course true for all $x^n\neq x'^n\in\cC_t$, so we can take $\lambda_2=2^{-n^{\gamma}t/4}$, which vanishes for increasing $n$. Both errors are converging to zero, so $\cC_t$ is a good DI code. It is only left to bound its size. We can follow again the combinatorial techniques in the proof of Theorem \ref{thm:cq-DI-capacity}, the size of $\cC_t$ is bounded in Eq.~\eqref{eq:size_Ct}:
\[
N=\lvert{\cC_t}\rvert 
\geq 2^{-n} \lvert{\cX_{n^{-\alpha}}}\rvert^{n(1-t)}
\geq 2^{-n}(Kn^\frac{1-\gamma}{2})^{n(d-\epsilon)(1-t)}.
\]
Taking the logarithm, we get:
\[
\log N\geq n\left(\frac{1-\gamma}{2}\right)(1-t)(d-\epsilon)\log n+o(n\log n).
\]
Finally, as we can choose the parameters $\gamma, t, \epsilon>0$ arbitrarily small, the linearithmic DI capacity is indeed lower bounded by 
\[
\dot{C}_\text{DI}(W)
    \geq \liminf_{n\rightarrow\infty}\frac{1}{n\log n} \log N
    \geq \frac12 \underline{d}_M\!\left(\!\sqrt{\!\widetilde{\cX}}\right)
\]
As already discussed in previous sections, the general converse proof from \cite{CDBW:DI_classical} gives us $\dot{C}_\text{DI}(W)\leq\frac12 \underline{d}_M\!\left(\!\sqrt{\!\widetilde{\cX}}\right)$. Therefore, the upper and lower bounds match, meaning that
\[\dot{C}_\text{DI}(W)=\frac12 \underline{d}_M\!\left(\!\sqrt{\!\widetilde{\cX}}\right)=\frac12\underline{d}_M\!\left(\widetilde{\cX}\right).\]
The last equality following from the fact that the square root map on pure states does not affect the dimension.
\end{proof}

\section{Conclusions and discussion}
\label{sec:conclusions}
In the present contribution we provide a quantum generalisation of the main tool used to construct general codes for deterministic identification, the Hypothesis Testing Lemma \cite{CDBW:DI_classical}. Before this result, some codes and capacity lower bounds for the quantum case had been calculated through the use of a quantum measurement that effectively reduced the quantum case to a classical problem where the classical Hypothesis Testing Lemma could be applied. This trick, however, restricted the construction of the quantum channel DI codes to have simultaneous decoders. With the new Quantum Hypothesis Testing Lemma \ref{lemma:qHTL} this restriction can be lifted, and we can construct directly general codes for DI over cq-channels and quantum channels with product states as input. We not only provide a code construction (Theorem \ref{thm:qDI_construction}), but we also analyse the capacities it can achieve. The unrestricted quantum DI capacity, previously lower bounded by the simultaneous capacity, can now be tightened and written in terms of the Minkowski dimension of the square rooted space of output quantum states (Theorem \ref{thm:cq-DI-capacity}). Both lower and upper bounds are now written in terms of this object, highlighting its relevance and providing the intuition that the true value of the capacity might be proportional to it: $\dot{C}_{\text{DI}}(W)=\alpha\,\underline{d}_M\!\left(\!\sqrt{\!\widetilde{\cX}}\right)$ with $\frac12\leq\alpha\leq\frac14$. 
The results also extends to the optimistic DI capacity
(Theorem \ref{thm:optimistic}).

We also presented and example channel for which we can construct an optimum code. Its capacity is $\dot{C}_{\text{DI}}(W)=\frac12\underline{d}_M\!\left(\!\sqrt{\!\widetilde{\cX}}\right)$. This result gives strong intuition that this might be the actual capacity value in the general case.

These results draw our attention back to the simultaneous case, which now has lower and upper capacity bounds expressed in terms of different objects. Specifically, it is bounded from above by the unrestricted case, hence the following natural questions arise: can we develop a converse proof directly for the simultaneous case? And will the resulting upper bound be written in terms of some $d_M\!\left(\!\sqrt{\!\widehat{\cX}}\right)$, matching the simultaneous lower bound?
On the one hand, a result like this 
could be taken to suggest that there is a separation between the two flavours of DI capacity. 
On the other hand, however, while it is clear that for any channel 
\(
  d_M\!\left(\!\sqrt{\!\widehat{\cX}}\right) \leq d_M\!\left(\!\sqrt{\!\widetilde{\cX}}\right),
\)
we have no evidence yet of a channel with a separation between these objects, which would be a necessary prerequisite for a capacity separation. Among other difficulties, there is the problem of maximization of the dimension over all possible measurements. 

\section*{Acknowledgments}
{\small
P. Colomer, H. Boche and C. Deppe were supported in part by the German Federal Ministry of Research, Technology and Space (BMFTR) in the programme of ``Souverän. Digital. Vernetzt.'' within the project 6G-life, project no. 16KISK002 and 16KISK263. 
P. Colomer was supported in part by the BMBF Quantum Programm QD-CamNetz, grant 16KISQ077, QuaPhySI, grant 16KIS1598K, and QUIET, grant 16KISQ093.
P. Colomer and A. Winter are supported by the Institute for Advanced Study of the Technical University Munich.
H. Boche and C. Deppe are supported by the German Federal Ministry of Research, Technology and Space (BMFTR) within the national initiative on Post Shannon Communication (NewCom) under grant 16KIS1003K and 16KIS1005, within the national initiative QuaPhySI -- Quantum Physical Layer Service Integration under grant 16KISQ1598K and 16KIS2234 and within the national initiative ``QTOK -- Quantum tokens for secure authentication in theory and practice'' under grant 16KISQ038. They have also been supported by the initiative 6GQT (financed by the Bavarian State Ministry of Economic Affairs, Regional Development and Energy).

H. Boche has further received funding from the German Research Foundation (DFG) within Germany’s Excellence Strategy EXC-2092 -- 390781972. C. Deppe received further funding under grants 16KISQ028, 16KISR038, 16KISQ0170, and DE1915:/2-1.

A. Winter was supported by the European Commission QuantERA project ExTRaQT (Spanish MICIN grant PCI2022-132965), by the Spanish MICIN (project PID2022-141283NB-I00) with the support of FEDER funds, by the Spanish MICIN with funding from European Union NextGenerationEU (PRTR-C17.I1) 
and the Generalitat de Catalunya, by the Spanish MTDFP through the QUANTUM ENIA project: Quantum Spain, funded by the European Union NextGenerationEU within the framework of the ``Digital Spain 2026 Agenda'', 

\vspace{-2pt}
\noindent and by the Alexander von Humboldt Foundation.
}

\AtNextBibliography{\footnotesize}
\printbibliography

\end{document}